\documentclass[a4paper, 10pt]{iopart}
\pdfoutput=1

\expandafter\let\csname equation*\endcsname\relax
\expandafter\let\csname endequation*\endcsname\relax

\usepackage[utf8]{inputenc}

\usepackage[yyyymmdd]{datetime} 

\usepackage{enumerate} 
\usepackage{enumitem} 

\usepackage{csquotes} 
\usepackage[numbers]{natbib}

\usepackage{braket} 
\usepackage{tensor} 

\usepackage{amsmath}
\usepackage{amsfonts}
\usepackage{amsthm}
\usepackage{amssymb}
\usepackage{mathrsfs}

\newcounter{mnotecount}[section]

\usepackage{xstring}
\usepackage[]{marginnote} 
\newcommand{\todo}[2][]{%
	\IfStrEqCase{#1}{%
        {inline}{{\color{red}#2}}%
	}[%
		\marginnote{\color{red}#2}%
	]
}

\usepackage{hyperref}								
\hypersetup{final, colorlinks, citecolor=blue, filecolor=blue, linkcolor=blue, urlcolor=blue}

\usepackage[noabbrev]{cleveref}
\crefname{equation}{}{}
\iflanguage{swedish}{ 
	
}{}
\iflanguage{english}{
	
}{}

\usepackage{orcidlink}

\usepackage[]{xurl}

\iflanguage{english}{
    \newtheorem{theorem}{Theorem}
    \newtheorem*{theorem*}{Theorem}
    \newtheorem{proposition}[theorem]{Proposition}
    \newtheorem*{proposition*}{Proposition}
    
    \newtheorem{corollary*}{Corollary}
    \newtheorem{lemma}[theorem]{Lemma}
    \newtheorem*{lemma*}{Lemma}
    
    \newtheorem*{remark*}{Remark}
    \theoremstyle{definition}
    \newtheorem{definition}[theorem]{Definition}
    \newtheorem*{definition*}{Definition}
}{}

\usepackage{aligned-overset}

\def\complexes{\mathbb{C}}

\def\from{\colon} 
\def\definedas{\coloneqq}

\newcommand*{\conjugate}[1]{\mkern 1.0mu\overline{\mkern-1.0mu#1\mkern-1.0mu}\mkern 1.0mu} 

\newcommand*{\group}[2]{{\mathrm{\MakeUppercase{#1}}}{\left(#2\right)}}

\def\sDiv{\mathscr{D}}
\def\sTwist{\mathscr{T}}
\def\sCurl{\mathscr{C}}
\def\sCurlDagger{\mathscr{C}^\dagger}

\newcommand*{\SymMult}[2]{\overset{#1, #2}{\odot}}

\newcommand*{\coeff}[3]{\underset{#2, #3}{#1}{}}

\def\SymSpace{\mathcal{S}}

\def\symmetryOperator{L}

\begin{document} 

\title{Second order symmetry operators for the massive Dirac equation}
\author[S. Jacobsson]{Simon Jacobsson \orcidlink{0000-0002-1181-972X}}
\ead{simon.jacobsson@kuleuven.be}
\address{Chalmers University of Technology, SE-412~96 Gothenburg, Sweden\\
KU Leuven, BE-3000 Leuven, Belgium}
\author[T. B\"ackdahl]{Thomas B\"ackdahl \orcidlink{0000-0003-3240-2445}} 
\ead{thomas.backdahl@chalmers.se}
\address{Mathematical Sciences, Chalmers University of Technology and University of Gothenburg, SE-412~96 Gothenburg, Sweden}
\date{\today}

\section*{Abstract}

Employing the covariant language of two-spinors, we find what conditions a curved Lorentzian spacetime must satisfy for existence of a second order symmetry operator for the massive Dirac equation.
The conditions are formulated as existence of a set of Killing spinors satisfying a set of covariant linear differential equations.
Using these Killing spinors, we then state the most general form of such an operator.
Partial results for the zeroth and first order are presented and interpreted as well.
Computer algebra tools from the \emph{Mathematica} package suite \emph{xAct} were used for the calculations.

\section{Introduction}%
\label{cha:introduction}
A \emph{symmetry operator} is a linear differential operator mapping solutions to solutions of a differential equation.
Such operators can be very useful for detailed studies of the solutions. 
However, the existence of such operators is not trivial and is linked to the existence of different kinds of symmetries of the curved spacetime geometry that the differential equation is defined on. This paper aims to elucidate this for the massive Dirac equation.

Many partial differential equations from physics, such as the Schr\"{o}dinger and Helmholtz equations, lend themselves naturally to separation of variables, but also the Dirac equation has been separated in some cases. This is closely related to the existence of symmetry operators. For instance, Kalnins et al.~\cite[section 3]{KalMilWil92}, explain the separation of the Dirac equation on the Kerr spacetime in terms of the existence of symmetry operators associated with a Killing tensor by identifying a set of separation constants as eigenvalues of said symmetry operators. 

Symmetry operators also have other uses, for instance, given a conserved energy, or an energy estimate, one can easily construct higher order versions by inserting a symmetry operator. More advanced uses have also been found. Andersson and Blue~\cite{AndBlu15} used higher order symmetry operators for the scalar wave equation on the Kerr spacetime to handle the complicated trapping phenomena when proving decay estimates.

For many differential equations, a Lie derivative along a Killing vector gives a symmetry operator, i.e.\ a symmetry of the spacetime gives a symmetry operator. 
However, in many cases there are also other less obvious symmetries sometimes called hidden symmetries that can give rise to symmetry operators. In general these are described in terms of Killing spinors. An important example is a second order symmetry operator related to the \emph{Carter constant}~\cite{Car68} used by Andersson and Blue in \cite{AndBlu15}. This symmetry operator can not be built from Killing vectors. 

To know that all symmetry operators have been found, a systematic study is required. If the set of symmetry operators is not large enough, the methods described above will not give satisfactory results.

The conditions for existence of symmetry operators we present here are described as existence of a set of Killing spinors satisfying a set of covariant differential equations. This can be interpreted as conditions on the spacetime geometry.
Assuming the spacetime is a sufficiently smooth four-dimensional Lorentzian manifold that allows for a spin structure, these conditions are both necessary and sufficient.

The spin structure allow us to decompose tensorial objects into irreducible components.
Using the covariant two-spinor formalism described by Penrose and Rindler~\cite{PenRinVol1,PenRinVol2}, these decompositions are used to decompose equations into independent subequations that must be satisfied simultaneously. 

It is in general a time-consuming and nontrivial task to find these irreducible decompositions.
Thus, for this task, computer algebra systems such as the \emph{Mathematica} packages \emph{Sym\-Manipulator}~\cite{SymManipulatorWeb} and \emph{SymSpin}~\cite{SymSpinWeb} have been developed.
While there is considerable power in basic \emph{Mathematica}, \emph{SymManipulator} lets the user handle abstract symmetrized tensor expressions, and automatically decompose spinors into irreducible symmetric spinors. \emph{SymSpin} allows the user to handle complicated expressions with such spinors in an efficient way. 

\bigskip

The massive \emph{Dirac equation} is, in spinor form,
\begin{subequations}
\begin{align}
	\nabla{}^A{}_{A'} \phi_A &= m \chi{}_{A'}\,,%
	\label{eq:left_dirac_explicit}\\
	\nabla{}_A{}^{A'} \chi_{A'} &= -m \phi_A\,,%
	\label{eq:right_dirac_explicit}
\end{align}
\end{subequations}
where $\phi_A$ and $\chi_{A'}$ are spinor fields.
The mass $m$ is assumed to be nonzero.
The first result in this article is that there are no nontrivial zeroth order symmetry operators.
The second result is that there exists a first order symmetry operator if and only if there exist Killing spinors satisfying auxiliary condition A.

\begin{definition}\label{def:explicit_auxiliary_condition_A}
	Let $S_{AA'}$, $T_{A'B'}$, $U_{AB}$, and $R_{AA'}$ be Killing spinors on a Lorentzian manifold.
	They satisfy \emph{auxiliary condition A} if
	\begin{subequations}
		\begin{align}
			\nabla_{(A}{}^{A'}S_{B)A'}={}&0 \label{eq:order1AuxiliaryConditionsExplicit1},\\
			\nabla^{A}{}_{(A'}S_{|A|B')}={}&0 \label{eq:order1AuxiliaryConditionsExplicit2},\\
			\nabla_{AA'}R^{AA'}={}&0 \label{eq:order1AuxiliaryConditionsExplicit3},\\
			\nabla_{AB'}T_{A'}{}^{B'} + \nabla_{BA'}U_{A}{}^{B}={}&0 \label{eq:order1AuxiliaryConditionsExplicit4}.
		\end{align}
	\end{subequations}
\end{definition}

The third result is that there exists a second order symmetry operator if and only if there exist Killing spinors satisfying auxiliary condition B, which we will state later in \cref{def:auxiliary_condition_B} after some notation has been introduced.

\bigskip

For this article, we have used \emph{Mathematica} version 13.1.0, \emph{xTensor} version 1.2.0, \emph{Spinors} version 1.0.6, \emph{SymManipulator} version 0.9.5, \emph{SymSpin} version 0.1.1, and \emph{TexAct} version 0.4.3.
A notebook used for creating all of the results presented in the following sections are available in a GitHub repository~\cite[]{myGithubRepo}.

\subsection{Previous work}%
\label{sub:previous_work}

Michel, Radoux and \v{S}ilhan \cite{michel:radoux:silhan:2013arXiv1308.1046M} analysed the symmetry operators for the conformal wave equation.

In \cite{AndBaeBlu14a} a method was developed to find all second order symmetry operators for the conformal wave equation, the Dirac--Weyl equation, and the Maxwell equation.
Their results are also formulated as existence of a set of Killing spinors satisfying a set of covariant differential equations.
We use the same method here.
As we are dealing with a more complicated system of equations, we will however take advantage of the recent development of the \emph{SymSpin} package.

\bigskip

The conditions
\cref{eq:order1AuxiliaryConditionsExplicit1,%
eq:order1AuxiliaryConditionsExplicit2,%
eq:order1AuxiliaryConditionsExplicit3,%
eq:order1AuxiliaryConditionsExplicit4}
for the existence of a first order symmetry operator and the form of that operator, presented in \cref{thm:condition_for_1st_order_symmetry_operator_for_dirac_equation}, is a reformulation of a result by Kamran and McLenaghan~\cite[theorem II]{KamMcL84} into covariant spinor language.
Benn and Kress~\cite{BenKre03} have showed that this result is the most general one of the first order in the sense that it extends to arbitrary spin manifolds.

A special case of the second order symmetry operator presented in this article has been derived by Fels and Kamran~\cite[theorem 4.1]{FelKam90}.

Auxiliary condition A can be interpreted very geometrically.
In \cref{sub:first_order_symmetry_operator:interpretation_and_discussion}, we show that \cref{eq:order1AuxiliaryConditionsExplicit4} implies the existence of a Killing--Yano tensor field.
If an operator commutes with the Dirac operator, then it is a symmetry operator, and so the set of operators commutating with the Dirac equation is a subset of the symmetry operators.
Previous work has been able to relate such operators to Killing--Yano tensors~\cite[]{Car04,CarKrtKub11}.
But also general symmetry operators have been studied in terms of Killing--Yano tensors~\cite[]{BenCha97,AckErtOndVer09}.

\section{Preliminaries}%
\label{sec:preliminaries}

In this section, the notation and concepts used in this article are presented.
Abstract index notation~\cite[chapter 2]{PenRinVol1} is used throughout and conventions are consistent with Penrose and Rindler ~\cite[]{PenRinVol1,PenRinVol2}.
Lowercase latin letters are used for Lorentzian tensor indices while uppercase latin letters are used for spinorial tensor indices, with a prime to indicate indices in the conjugate space.

\subsection{Killing tensors}%
\label{sub:killing_tensors}

A Killing vector is a vector field $\tensor{K}{^c}$ such that taking the Lie derivative of the metric with respect to it is zero, which can be written as $\tensor{\nabla}{_{(a}} \tensor{K}{_{b)}} = 0$.
The following definitions are then natural generalizations,
\begin{definition}
	A vector $\tensor{K}{^c}$ is a \emph{conformal Killing vector} if
	\begin{align}
		\tensor{\nabla}{_{(a}} \tensor{K}{_{b)}} = \lambda \tensor{g}{_a_b}
	\end{align}
	for some scalar field $\lambda$.
\end{definition}
\begin{definition}\label{def:killing_tensor}
	A totally symmetric tensor $\tensor{K}{_{b \dots q}}$ is a \emph{Killing tensor} if
	\begin{align}\label{eq:killing_tensor}
		\tensor{\nabla}{_{(a}} \tensor{K}{_{b \dots q)}} = 0.
	\end{align}
\end{definition}
\begin{definition}\label{def:killing_spinor}
	A totally symmetric spinor $\tensor{S}{_{B \dots Q}^{B' \dots Q'}}$ is a \emph{Killing spinor} if
	\begin{align}\label{eq:killing_spinor}
		\tensor{\nabla}{_{(A}^{(A'}} \tensor{S}{_{B \dots Q)}^{B' \dots Q')}} = 0.
	\end{align}
\end{definition}

Another type of geometrical quantitity of interest is Killing--Yano tensors.
They are used to construct valence $2$ Killing tensors and sometimes they are easier to find than the Killing tensors they correspond to.
\begin{definition}
	A totally antisymmetric tensor $\tensor{f}{_{b_0 \dots b_n}}$ is a \emph{Killing--Yano tensor} if
	\begin{align}\label{eq:killing-yano}
		\tensor{\nabla}{_{(a}} \tensor{f}{_{b_0) b_1 \dots b_n}} = 0.
	\end{align}
\end{definition}

Lastly for this subsection, we will define the \emph{conformally weighted Lie derivative}~\cite[(15)]{AncPoh04}, \cite[(2.5)]{AndBaeBlu14a}.
It will be used to interpret some of the terms in the symmetry operators.
\begin{definition}
	If $\tensor{\xi}{_A^{A'}}$ is a Killing vector, and $\tensor{\varphi}{_{A_1 \dots A_{k}}}$ is a totally symmetric valence $(k, 0)$ spinor, then
	\begin{align}
		\hat{\mathcal{L}}_\xi \tensor{\varphi}{_{A_1 \dots A_{k}}} = \tensor{\xi}{^{B B'}} \tensor{\nabla}{_{B B'}} \tensor{\varphi}{_{A_1 \dots A_{k}}} + \frac{k}{2} \tensor{\varphi}{_{B (A_2 \dots A_{k}}} \tensor{\nabla}{_{A_1) B'}} \tensor{\xi}{^{B B'}} + \frac{2 - k}{8} \tensor{\varphi}{_{A_1 \dots A_{k}}} \tensor{\nabla}{^{B B'}} \tensor{\xi}{_{B B'}}.
	\end{align}
	If $\varphi$ is instead of valence $(0, k)$, then $\hat{\mathcal{L}}_\xi \varphi$ is defined as $\conjugate{\hat{\mathcal{L}}_\xi \conjugate{\varphi}}$.
\end{definition}

\subsection{Decomposing spinors}%
\label{sub:decomposing_spinors}

We formulated the Dirac equation in \cref{eq:left_dirac_explicit,eq:right_dirac_explicit} using two-spinors.
Two-spinors transform under the universal covering group, $\group{SL}{2, \complexes}$, of the proper Lorentz group.
Something that greatly simplifies discussions about two-spinors is that, when working over $\group{SL}{2, \complexes}$, the only spinorial tensor that is antisymmetric in more than two indices is $0$ and the only spinorial tensor antisymmetric in two indices is the spin-metric $\tensor{\epsilon}{_A_B}$ and its multiples.
From this follows a very useful result, proved in Penrose and Rindler~\cite[proposition 3.3.54]{PenRinVol1}.
\begin{theorem}\label{thm:spinor_decomposition}
	Any spinor $T_{A_1 \dots A_p}{}^{A'_1 \dots A'_q}$ is the sum of $T_{(A_1 \dots A_p)}{}^{(A'_1 \dots A'_q)}$ and linear combinations of outer products of symmetric spinors of lower valence with spin-metrics.
\end{theorem}

As an example of this theorem, the spinorial Riemann tensor, $\tensor{R}{_A_{A'}_B_{B'}_C_{C'}_D_{D'}}$, can be decomposed as~\cite[(13.2.25)]{Wal84} 
\begin{align}
	\tensor{R}{_{A A' B B' C C' D D'}} ={}& \tensor{\Psi}{_{A B C D}} \tensor{\epsilon}{_{A' B'}} \tensor{\epsilon}{_{C' D'}}
	+ \Lambda \left( \tensor{\epsilon}{_{A C}} \tensor{\epsilon}{_{B D}} + \tensor{\epsilon}{_{B C}} \tensor{\epsilon}{_{A D}} \right) \tensor{\conjugate \epsilon}{_{A' B'}} \tensor{\conjugate \epsilon}{_{C' D'}}\nonumber\\
	&+ \tensor{\Phi}{_{A B C' D'}} \tensor{\conjugate \epsilon}{_{A' B'}} \tensor{\epsilon}{_{C D}} + \textrm{complex conjugate}.%
	\label{eq:spinorial_riemann_tensor_decomposed_furthest}
\end{align}
$\tensor{\Psi}{_{A B C D}} = \frac{1}{4} \tensor{R}{_{(A B C D)}_{X'}^{X'}_{Y'}^{Y'}}$ is the \emph{Weyl spinor}, $\Lambda = \frac{1}{24} \tensor{R}{_X^Y_{X'}^{X'}_Y^X_{Y'}^{Y'}}$ is the \emph{Ricci scalar}, and $\tensor{\Phi}{_{A B C' D'}} = \frac{1}{4} \tensor{R}{_{(A B) X'}^{X'}_X^X_{(C' D')}}$ is the \emph{Ricci spinor}.

\subsection{Index-free notation}%
\label{sub:index_free_notation}

\Cref{thm:spinor_decomposition} allows us to decompose spinors into sums of outer products of symmetric spinors and $\epsilon$:s, but if an expression is symmetric in all of its free indices, then, after applying \cref{thm:spinor_decomposition}, every $\epsilon$ will have at least one index contracted.
So the expression may be written only in terms of partially contracted outer products of symmetric spinors.
If two symmetric spinors are multiplied and partially contracted, it does not matter \emph{which} indices are contracted, only how many.

Hence the following definition.
\begin{definition}[\cite{SymSpinArticle} definition 1]
	Let $\tensor{T}{_{A_1 \dots A_k}^{A_1' \dots A_l'}}$ and $\tensor{S}{_{A_1 \dots A_n}^{A_1' \dots A_n'}}$ be totally symmetric spinors.
	Then the \emph{symmetric $pq$-multiplication} of $\tensor{T}{_{A_1 \dots A_k}^{A_1' \dots A_l'}}$ with $\tensor{S}{_{A_1 \dots A_n}^{A_1' \dots A_n'}}$ is the totally symmetrized outer product of $T$ and $S$ where $p$ unprimed indices are contracted and $q$ primed indices are contracted:
	\begin{align}
		&\tensor{(T \SymMult{p}{q} S)}{_{A_1}_\dots_{A_{k + m - 2p}}^{A'_1}^\dots^{A'_{l + n - 2q}}}\nonumber\\
		&=
		\tensor{T}{_{(A_1}_\dots_{A_{k - p}}^{B_1}^\dots^{B_p}^{(A'_1}^\dots^{A'_{l - q}}^{|B'_1 \dots B'_{q}|}}
		\tensor{S}{_{A_{k - p + 1}}_\dots_{A_{k + m - 2 p})}_{B_1}_\dots_{B_p}^{A'_{l - q + 1}}^\dots^{A'_{l + n - 2q})}_{B'_1 \dots B'_{q}}}.
	\end{align}
\end{definition}

With this operator, we don't need to write out the indices in partially contracted outer products of symmetric spinors.
We will call this \emph{index-free notation}.

\subsection{Fundamental derivatives}%
\label{sub:fundamental_operators}

Another application of \cref{thm:spinor_decomposition} is to the covariant spinor derivative of a totally symmetric spinor.
Such an expression has four irreducible parts and we will name them as follows.
\begin{definition}[\cite{AndBaeBlu14a} definition 13]\label{def:fundamental_derivatives}
	Let $\mathcal{S}_{k, l}$ denote the space of smooth symmetric spinor fields of valence $(k, l)$ and let $\psi{}_{A_1 \dots A_k}{}^{A'_1 \dots A'_l} \in \mathcal{S}_{k, l}$.
	Then there are four \emph{fundamental derivatives}: the \emph{divergence} $\sDiv \from \mathcal{S}_{k, l} \to \mathcal{S}_{k - 1, l - 1}$ which acts by
	\begin{subequations}
	\begin{align}
		\tensor{(\sDiv \psi)}{_{A_1}_\dots_{A_{k-1}}^{A'_1}^\dots^{A'_{l-1}}} &= \tensor{\nabla}{^B^{B'}} \tensor{\psi}{_{A_1}_\dots_{A_{k-1}}_B^{A'_1}^\dots^{A'_{l-1}}_{B'}}
		\quad \text{for} \quad k \geq 1, l \geq 1,
	\end{align}
	the \emph{curl} $\sCurl \from \mathcal{S}_{k, l} \to \mathcal{S}_{k-1, l+1}$ which acts by
	\begin{align}
		\tensor{(\sCurl \psi)}{_{A_1}_\dots_{A_{k+1}}^{A'_1}^\dots^{A'_{l-1}}} &= \tensor{\nabla}{_{(A_1}^{B'}} \tensor{\psi}{_{A_2}_\dots_{A_{k+1})}^{A'_1}^\dots^{A'_{l-1}}_{B'}}
		\quad \text{for} \quad k \geq 0, l \geq 1,
	\end{align}
	the \emph{curl-dagger} $\sCurlDagger \from \mathcal{S}_{k, l} \to \mathcal{S}_{k+1, l-1}$ which acts by
	\begin{align}
		\tensor{(\sCurlDagger \psi)}{_{A_1}_\dots_{A_{k-1}}^{A'_1}^\dots^{A'_{l+1}}} &= \tensor{\nabla}{^B^{(A'_1}} \tensor{\psi}{_{A_1}_\dots_{A_{k-1}}_{B}^{A'_2}^\dots^{A'_{l+1})}}
		\quad \text{for} \quad k \geq 1, l \geq 0,
	\end{align}
	and the \emph{twist} $\sTwist \from \mathcal{S}_{k, l} \to \mathcal{S}_{k+1, l+1}$ which acts by
	\begin{align}
		\tensor{(\sTwist \psi)}{_{A_1}_\dots_{A_{k+1}}^{A'_1}^\dots^{A'_{l+1}}} &= \tensor{\nabla}{_{(A_1}^{(A'_1}} \tensor{\psi}{_{A_2}_\dots_{A_{k+1})}^{A'_2}^\dots^{A'_{l+1})}}
		\quad \text{for} \quad k \geq 0, l \geq 0.
	\end{align}
	\end{subequations}
\end{definition}
To make precise the statement that the fundamental derivatives are the irreducible parts of the spinor derivative, we state the following lemma.
\begin{lemma}[\cite{AndBaeBlu14a} lemma 15]\label{lemma:spinor_derivative_decomposition}
	Let $\tensor{\psi}{_{A_1 \dots A_k}^{A'_1 \dots A'_l}}$ be totally symmetric. Then
	\begin{align}
		\tensor{\nabla}{_{A_1}^{A'_1}} \tensor{\psi}{_{A_2}_\dots_{A_{k+1}}^{A'_2}^\dots^{A'_{l+1}}} ={}& \tensor{(\sTwist \psi)}{_{A_1}_\dots_{A_{k+1}}^{A'_1}^\dots^{A'_{l+1}}}\nonumber\\
		&- \frac{l}{l+1} \tensor{\conjugate \epsilon}{^{A'_1}^{(A'_2}} \tensor{(\sCurl \psi)}{_{A_1}_\dots_{A_{k+1}}^{A'_3}^\dots^{A'_{l+1})}}\nonumber\\
		&- \frac{k}{k+1} \tensor{\epsilon}{_{A_1}_{(A_2}} \tensor{(\sCurlDagger \psi)}{_{A_3}_\dots_{A_{k+1})}^{A'_1}^\dots^{A'_{k+1}}}\nonumber\\
		&+ \frac{kl}{(k + 1) (l + 1)} \tensor{\epsilon}{_{A_1}_{(A_2}} \tensor{\conjugate \epsilon}{^{A'_1}^{(A'_2}} \tensor{(\sDiv \psi)}{_{A_3}_\dots_{A_{k+1})}^{A'_3}^\dots^{A'_{l+1})}}.%
		\label{eq:covariant_derivative_of_symmetric_spinor_decomposition}
	\end{align}
\end{lemma}

The spinorial Bianchi identity may be formulated in terms of fundamental derivatives.
\begin{lemma}\label{lemma:bianchi_identity}
	The Bianchi identity for the spinorial Riemann tensor is
	\begin{subequations}
	\begin{align}
		\tensor{(\sDiv \Phi)}{_A^{A'}} ={}& -3 \tensor{(\sTwist \Lambda)}{_A^{A'}},\\
		\tensor{(\sCurlDagger \Psi)}{_A_B_C^{A'}} ={}& \tensor{(\sCurl \Phi)}{_A_B_C^{A'}}.
	\end{align}
	\end{subequations}
\end{lemma}
We will use this identity along with its complex conjugate to simplify and canonicalize expressions containing derivatives of the spinorial Riemann tensor.

Another observation that will later form the bridge between spinor algebra and spacetime geometry is that \cref{def:killing_spinor} may be reformulated as
\begin{proposition}\label{prop:killing_spinor_condition}
	A totally symmetric valence $(k, l)$ spinor $\tensor{\psi}{_B_\dots_Q^{B'}^\dots^{Q'}}$ is a Killing spinor if and only if
	\begin{align}
		\tensor{(\sTwist S)}{_A_B_\dots_Q^{A'}^{B'}^\dots^{Q'}} = 0.
	\end{align}
\end{proposition}

\subsection{Commutators of fundamental derivatives}%
\label{sub:commutators}

\begin{definition}
	The \emph{spinor box operators} are
	\begin{subequations}
	\begin{align}
		\tensor{\square}{_A_B} &= \tensor{\nabla}{_{(A}_{|A'|}} \tensor{\nabla}{_{B)}^{A'}},\\
		\tensor{\conjugate\square}{_{A'}_{B'}} &= \tensor{\nabla}{_{(A}_{|A'|}} \tensor{\nabla}{^A_{B)}}.
	\end{align}
	\end{subequations}
\end{definition}

Note that both are contractions of the expression
\begin{align*}
	\tensor{\nabla}{_A_{A'}} \tensor{\nabla}{_B_{B'}} - \tensor{\nabla}{_B_{B'}} \tensor{\nabla}{_A_{A'}}.
\end{align*}
Hence any box operator acting on a spinor may be re-expressed as some partial contraction between that spinor and the spinorial Riemann tensor.
Importantly, the spinor box operators can be rewritten to be order $0$ in derivative:
\begin{lemma}
	Let $\psi$ be a valence $(k, l)$ spinor.
	Then
	\begin{subequations}
	\begin{align}
		\square \SymMult{0}{0} \psi ={}& -k \Psi \SymMult{1}{0} \psi - l \Phi \SymMult{0}{1} \psi,\\
		\square \SymMult{1}{0} \psi ={}& -(k - 1) \Psi \SymMult{2}{0} \psi - l \Phi \SymMult{1}{1} \psi - (k + 2) \Lambda \SymMult{0}{0} \psi,\\
		\square \SymMult{2}{0} \psi ={}& -(k - 2) \Psi \SymMult{3}{0} \psi - l \Phi \SymMult{2}{1} \psi.
	\end{align}
	\end{subequations}
\end{lemma}

Box operators appear when commuting fundamental derivatives.
\begin{lemma}[\cite{AndBaeBlu14a} lemma 18]
	Let $\psi$ be a valence $(k, l)$ spinor.
	Then the fundamental derivatives satisfy the following relations
	\begingroup
	\allowdisplaybreaks
	\begin{subequations} 
	\begin{align}
		\sDiv \sCurl \psi ={}&
		\frac{k}{k + 1} \sCurl \sDiv \psi
		- \conjugate\square \SymMult{0}{2} \psi,%
		\label{eq:div_curl_to_curl_div}\\
		\sDiv \sCurlDagger \psi ={}&
		\frac{k}{l + 1} \sCurlDagger \sDiv \psi
		- \square \SymMult{2}{0} \psi%
		\label{eq:div_curldagger_to_curldagger_div}\\
		\sCurl \sTwist \psi ={}&
		\frac{l}{l + 1} \sTwist \sCurl \psi
		- \square \SymMult{0}{0} \psi%
		\label{eq:curl_twist_to_twist_curl}\\
		\sCurlDagger \sTwist \psi ={}&
		\frac{k}{k + 1} \sTwist \sCurlDagger \psi
		- \conjugate\square \SymMult{0}{0} \psi%
		\label{eq:curldagger_twist_to_twist_curldagger}\\
		\sDiv \sTwist \psi ={}&
		-\left( \frac{1}{k + 1} + \frac{1}{l + 1} \right) \sCurl \sCurlDagger \psi
		+ \frac{l (l + 2)}{(l + 1)^2} \sTwist \sDiv \psi
		- \frac{l + 2}{l + 1} \square \SymMult{1}{0} \psi
		- \frac{l}{l + 1} \conjugate\square \SymMult{0}{1} \psi%
		\label{eq:div_twist_to_curl_curldagger_and_twist_div}\\
		\sDiv \sTwist \psi ={}&
		-\left( \frac{1}{k + 1} + \frac{1}{l + 1} \right) \sCurlDagger \sCurl \psi
		+ \frac{k (k + 2)}{(k + 1)^2} \sTwist \sDiv \psi
		- \frac{k + 2}{k + 1} \conjugate\square \SymMult{0}{1} \psi
		- \frac{k}{k + 1} \square \SymMult{1}{0} \psi%
		\label{eq:div_twist_to_curldagger_curl_and_twist_div}\\
		\sCurl \sCurlDagger \psi ={}&
		\sCurlDagger \sCurl \psi
		+ \left( \frac{1}{k + 1} - \frac{1}{l + 1} \right) \sTwist \sDiv \psi
		- \square \SymMult{1}{0} \psi
		+ \conjugate\square \SymMult{0}{1} \psi%
		\label{eq:curl_curldagger_to_curldagger_curl_and_twist_div}
	\end{align}
	\end{subequations}
	\endgroup
\end{lemma}

\subsection{Leibniz rules for fundamental derivatives}%
\label{sub:leibniz_rules_for_fundamental_derivatives}

The following lemma is formulated and proved by Aksteiner and B\"{a}ckdahl.
\begin{lemma}[\cite{SymSpinArticle} lemma 10]\label{lemma:symmult_liebniz_rules}
For symmetric spinors $\phi \in \SymSpace_{i,j}, \varphi \in \mathcal{S}_{k,l}$ we have the following Leibniz rules.
{\allowdisplaybreaks%
\begin{subequations}
\label{eq:SymMultLeibniz}
\begin{align}%
	\sTwist (\phi {\overset{m,n}{\odot}}\varphi)={}&(-1)^{m + n}\varphi {\overset{m,n}{\odot}}\sTwist \phi
	 + \tfrac{(-1)^{m + n} n}{j + 1}\varphi {\overset{m,n - 1}{\odot}}\sCurl \phi
	 + \tfrac{(-1)^{m + n} m}{i + 1}\varphi {\overset{m - 1,n}{\odot}}\sCurlDagger \phi\nonumber\\*
	& + \tfrac{(-1)^{m + n} m n}{(i + 1) (j + 1)}\varphi {\overset{m - 1,n - 1}{\odot}}\sDiv \phi
	 + \phi {\overset{m,n}{\odot}}\sTwist \varphi
	 + \tfrac{n}{l + 1}\phi {\overset{m,n - 1}{\odot}}\sCurl \varphi\nonumber\\*
	& + \tfrac{m}{k + 1}\phi {\overset{m - 1,n}{\odot}}\sCurlDagger \varphi
	 + \tfrac{m n}{k l + k + l + 1}\phi {\overset{m - 1,n - 1}{\odot}}\sDiv \varphi 
	\label{eq:twist_product_rule},\\
	\sCurl (\phi {\overset{m,n}{\odot}}\varphi)={}&\tfrac{(-1)^{m + n + 1} (l -  n)}{j + l - 2 n}\varphi {\overset{m,n + 1}{\odot}}\sTwist \phi
	 + \tfrac{(-1)^{m + n} (j -  n) (j + l -  n + 1)}{(j + 1) (j + l - 2 n)}\varphi {\overset{m,n}{\odot}}\sCurl \phi\nonumber\\*
	& + \tfrac{(-1)^{m + n + 1} m (l -  n)}{(i + 1) (j + l - 2 n)}\varphi {\overset{m - 1,n + 1}{\odot}}\sCurlDagger \phi\nonumber\\*
	& + \tfrac{(-1)^{m + n} m (j -  n) (j + l -  n + 1)}{(i + 1) (j + 1) (j + l - 2 n)}\varphi {\overset{m - 1,n}{\odot}}\sDiv \phi
	 -  \tfrac{j -  n}{j + l - 2 n}\phi {\overset{m,n + 1}{\odot}}\sTwist \varphi\nonumber\\*
	& + \tfrac{(l -  n) (j + l -  n + 1)}{(l + 1) (j + l - 2 n)}\phi {\overset{m,n}{\odot}}\sCurl \varphi
	 + \tfrac{m (- j + n)}{(k + 1) (j + l - 2 n)}\phi {\overset{m - 1,n + 1}{\odot}}\sCurlDagger \varphi\nonumber\\*
	& + \tfrac{m (l -  n) (j + l -  n + 1)}{(k + 1) (l + 1) (j + l - 2 n)}\phi {\overset{m - 1,n}{\odot}}\sDiv \varphi 
	\label{eq:curl_product_rule},\\
	\sCurlDagger (\phi {\overset{m,n}{\odot}}\varphi)={}&\tfrac{(-1)^{m + n + 1} (k -  m)}{i + k - 2 m}\varphi {\overset{m + 1,n}{\odot}}\sTwist \phi
	 + \tfrac{(-1)^{m + n + 1} n (k -  m)}{(j + 1) (i + k - 2 m)}\varphi {\overset{m + 1,n - 1}{\odot}}\sCurl \phi\nonumber\\*
	& + \tfrac{(-1)^{m + n} (i -  m) (i + k -  m + 1)}{(i + 1) (i + k - 2 m)}\varphi {\overset{m,n}{\odot}}\sCurlDagger \phi\nonumber\\*
	& + \tfrac{(-1)^{m + n} n (i -  m) (i + k -  m + 1)}{(i + 1) (j + 1) (i + k - 2 m)}\varphi {\overset{m,n - 1}{\odot}}\sDiv \phi
	 -  \tfrac{i -  m}{i + k - 2 m}\phi {\overset{m + 1,n}{\odot}}\sTwist \varphi\nonumber\\*
	& + \tfrac{n (- i + m)}{(l + 1) (i + k - 2 m)}\phi {\overset{m + 1,n - 1}{\odot}}\sCurl \varphi
	 + \tfrac{(k -  m) (i + k -  m + 1)}{(k + 1) (i + k - 2 m)}\phi {\overset{m,n}{\odot}}\sCurlDagger \varphi\nonumber\\*
	& + \tfrac{n (k -  m) (i + k -  m + 1)}{(k + 1) (l + 1) (i + k - 2 m)}\phi {\overset{m,n - 1}{\odot}}\sDiv \varphi 
	\label{eq:curl_dagger_product_rule},\\
	\sDiv (\phi {\overset{m,n}{\odot}}\varphi)={}&\tfrac{(-1)^{m + n} (k -  m) (l -  n)}{(i + k - 2 m) (j + l - 2 n)}\varphi {\overset{m + 1,n + 1}{\odot}}\sTwist \phi\nonumber\\*
	& + \tfrac{(-1)^{m + n + 1} (j -  n) (k -  m) (j + l -  n + 1)}{(j + 1) (i + k - 2 m) (j + l - 2 n)}\varphi {\overset{m + 1,n}{\odot}}\sCurl \phi\nonumber\\*
	& + \tfrac{(-1)^{m + n + 1} (i -  m) (l -  n) (i + k -  m + 1)}{(i + 1) (i + k - 2 m) (j + l - 2 n)}\varphi {\overset{m,n + 1}{\odot}}\sCurlDagger \phi\nonumber\\*
	& + \tfrac{(-1)^{m + n} (i -  m) (j -  n) (i + k -  m + 1) (j + l -  n + 1)}{(i + 1) (j + 1) (i + k - 2 m) (j + l - 2 n)}\varphi {\overset{m,n}{\odot}}\sDiv \phi\nonumber\\*
	& + \tfrac{(i -  m) (j -  n)}{(i + k - 2 m) (j + l - 2 n)}\phi {\overset{m + 1,n + 1}{\odot}}\sTwist \varphi\nonumber\\*
	& + \tfrac{(- i + m) (l -  n) (j + l -  n + 1)}{(l + 1) (i + k - 2 m) (j + l - 2 n)}\phi {\overset{m + 1,n}{\odot}}\sCurl \varphi\nonumber\\*
	& + \tfrac{(j -  n) (- k + m) (i + k -  m + 1)}{(k + 1) (i + k - 2 m) (j + l - 2 n)}\phi {\overset{m,n + 1}{\odot}}\sCurlDagger \varphi\nonumber\\*
	& + \tfrac{(k -  m) (l -  n) (i + k -  m + 1) (j + l -  n + 1)}{(k + 1) (l + 1) (i + k - 2 m) (j + l - 2 n)}\phi {\overset{m,n}{\odot}}\sDiv \varphi 
	\label{eq:div_product_rule}.
\end{align}
\end{subequations}
}
\end{lemma}

\subsection{Reduced ansatz}%
\label{sub:writing_symmetry_operators_in_terms_of_twists}
The Dirac equation is 
\begin{subequations}
\begin{align}
	\tensor{(\sCurl^\dagger \phi)}{_{A'}} &= m \tensor{\chi}{_{A'}},%
	\label{eq:left_dirac}\\
	\tensor{(\sCurl \chi)}{_A} &= -m \tensor{\phi}{_A}%
	\label{eq:right_dirac}.
\end{align}
\end{subequations}
The condition that a differential operator $\symmetryOperator \from (\phi_A, \chi_{A'}) \mapsto (\lambda_A, \gamma_{A'})$ is a symmetry operator for the Dirac equation is that
\begin{subequations}
\begin{align}
	\tensor{(\sCurlDagger \lambda)}{_{A'}} &= m \tensor{\gamma}{_{A'}}, 
	\label{eq:left_dirac_for_lambda}\\
	\tensor{(\sCurl \gamma)}{_A} &= -m \tensor{\lambda}{_A}.
	\label{eq:right_dirac_for_gamma}
\end{align}
\end{subequations}
for all $(\phi_A, \chi_{A'})$ satisfying \cref{eq:left_dirac,eq:right_dirac}.

\begin{lemma}\label{lemma:differential_operator_expressible_as_twists}
	Any symmetry operator $\symmetryOperator$ of the Dirac equation may be written only in terms of twists.
\end{lemma}

\begin{proof}
	We will show this by induction on the order of the differential operator.

	For the base case, consider that by \cref{lemma:spinor_derivative_decomposition},
	\begin{align}
		\tensor{\nabla}{_A^{A'}} \phi_B &= \tensor{(\sTwist \phi)}{_A_B^{A'}} - \frac{1}{2} \tensor{\epsilon}{_A_B}\tensor{(\sCurlDagger \phi)}{^{A'}}\nonumber\\
	\overset{\text{\cref{eq:left_dirac}}}&{=} \tensor{(\sTwist \phi)}{_A_B^{A'}} - \frac{1}{2} \tensor{\epsilon}{_A_B} m \tensor{\chi}{^{A'}} 
	\end{align}
	and
	\begin{align}
		\tensor{\nabla}{_A^{A'}} \chi^{B'} &= \tensor{(\sTwist \chi)}{_A^{A'}^{B'}} - \frac{1}{2} \tensor{\epsilon}{^{A'}^{B'}}\tensor{(\sCurl \chi)}{_A}\nonumber\\
	\overset{\text{\cref{eq:right_dirac}}}&{=} \tensor{(\sTwist \chi)}{_A^{A'}^{B'}} + \frac{1}{2} \tensor{\epsilon}{^{A'}^{B'}} m \tensor{\phi}{_A}.
	\end{align}

	For the induction step, we need only consider three cases.
	Let $\#$ stand for \enquote{some coefficient}, $S$ for either $\phi$ or $\chi$, and $H$ for the induction hypothesis.
	We will also use that whenever a spinor box operator appears we may write it as a partial contraction with the Riemann spinor.
	Then
	{\allowdisplaybreaks
	\begin{align}
		\sCurl \underbrace{\sTwist \dots \sTwist}_{\times n} S \overset{\text{\cref{eq:curl_twist_to_twist_curl}}}&{=} {\#} \sTwist \sCurl \underbrace{\sTwist \dots \sTwist}_{\times (n - 1)} S + {\#} \square \underbrace{\sTwist \dots \sTwist}_{\times (n - 1)} S\nonumber\\*
		&{}= {\#} \sTwist \sCurl \underbrace{\sTwist \dots \sTwist}_{\times (n - 1)} S + \text{lower order terms}\nonumber\\*
		\overset{H}&{=} {\#} \underbrace{\sTwist \dots \sTwist}_{\times n} S + \text{lower order terms},%
		\label{eq:curl_twists_to_twists}\\
		\sCurlDagger \underbrace{\sTwist \dots \sTwist}_{\times n} S \overset{\text{\cref{eq:curldagger_twist_to_twist_curldagger}}}&{=} \# \sTwist \sCurlDagger \underbrace{\sTwist \dots \sTwist}_{\times (n - 1)} S + {\#} \square \underbrace{\sTwist \dots \sTwist}_{\times (n - 1)} S\nonumber\\*
		&{}= {\#} \sTwist \sCurlDagger \underbrace{\sTwist \dots \sTwist}_{\times (n - 1)} S + \text{lower order terms}\nonumber\\*
		\overset{H}&{=} {\#} \underbrace{\sTwist \dots \sTwist}_{\times n} S + \text{lower order terms},%
		\label{eq:curldagger_twists_to_twists}\\
		\sDiv \underbrace{\sTwist \dots \sTwist}_{\times n} S \overset{\text{\cref{eq:div_twist_to_curl_curldagger_and_twist_div}}}&{=} \# \sCurl \sCurlDagger \underbrace{\sTwist \dots \sTwist}_{\times (n - 1)} S + \# \sTwist \sDiv \underbrace{\sTwist \dots \sTwist}_{\times (n - 1)} S + \# \square \underbrace{\sTwist \dots \sTwist}_{\times (n - 1)} S\nonumber\\*
		&{}=  \# \sCurl \sCurlDagger \underbrace{\sTwist \dots \sTwist}_{\times (n - 1)} S + \# \sTwist \sDiv \underbrace{\sTwist \dots \sTwist}_{\times (n - 1)} S + \text{lower order terms}\nonumber\\*
		\overset{H}&{=} \# \sCurl \underbrace{\sTwist \dots \sTwist}_{\times n} S + \# \underbrace{\sTwist \dots \sTwist}_{\times n} S + \text{lower order terms}\nonumber\\*
		\overset{\text{\cref{eq:curl_twists_to_twists}}}&{=} \# \underbrace{\sTwist \dots \sTwist}_{\times n} S + \text{lower order terms}.%
		\label{eq:div_twists_to_twists}
	\end{align}
	}%
	Note that the right-most sides of \cref{eq:curl_twists_to_twists,eq:curldagger_twists_to_twists,eq:div_twists_to_twists} all have one less order than the left-most sides.
\end{proof}

This means that the only derivative operator we need in an ansatz for a symmetry operator is the twist operator.
It is to great advantage that the proof is constructive.
It allows the first orders to be calculated explicitly.
Order one was shown as the base case.
The second order comes out to
{\allowdisplaybreaks
\begin{subequations}
	\begin{align}
\sDiv \sTwist \phi={}&(\tfrac{3}{2} m^2 - 6 \Lambda)\phi \label{eq:order2ReductionToTwists1},\\
\sDiv \sTwist \chi={}&(\tfrac{3}{2} m^2 - 6 \Lambda)\chi \label{eq:order2ReductionToTwists2},\\
\sCurl \sTwist \phi={}&\Psi {\overset{1,0}{\odot}}\phi \label{eq:order2ReductionToTwists3},\\
\sCurl \sTwist \chi={}& {-} \tfrac{1}{2} m\sTwist \phi
 + \Phi {\overset{0,1}{\odot}}\chi  \label{eq:order2ReductionToTwists4},\\
\sCurlDagger \sTwist \phi={}&\tfrac{1}{2} m\sTwist \chi
 + \Phi {\overset{1,0}{\odot}}\phi  \label{eq:order2ReductionToTwists5},\\
\sCurlDagger \sTwist \chi={}&\bar\Psi{\overset{0,1}{\odot}}\chi \label{eq:order2ReductionToTwists6},
\end{align}
\end{subequations}
}%
and the third order comes out to
{\allowdisplaybreaks
\begin{subequations}
\begin{align}
\sDiv \sTwist \sTwist \phi={}&(\tfrac{4}{3} m^2 - 12 \Lambda)\sTwist \phi
 -  \tfrac{5}{6}(\sCurl \Phi){\overset{1,0}{\odot}}\phi
 + \tfrac{5}{18}(\sDiv \Phi){\overset{0,0}{\odot}}\phi
 + \tfrac{10}{3}\Phi {\overset{1,1}{\odot}}\sTwist \phi
 -  \tfrac{9}{2}(\sTwist \Lambda){\overset{0,0}{\odot}}\phi\nonumber\\
& + \tfrac{3}{2}\Psi {\overset{2,0}{\odot}}\sTwist \phi  \label{eq:order3ReductionToTwists1},\\
\sDiv \sTwist \sTwist \chi={}&(\tfrac{4}{3} m^2 - 12 \Lambda)\sTwist \chi
 -  \tfrac{5}{6}(\sCurlDagger \Phi){\overset{0,1}{\odot}}\chi
 + \tfrac{5}{18}(\sDiv \Phi){\overset{0,0}{\odot}}\chi
 + \tfrac{10}{3}\Phi {\overset{1,1}{\odot}}\sTwist \chi\nonumber\\
& -  \tfrac{9}{2}(\sTwist \Lambda){\overset{0,0}{\odot}}\chi
 + \tfrac{3}{2}\bar\Psi{\overset{0,2}{\odot}}\sTwist \chi  \label{eq:order3ReductionToTwists2},\\
\sCurl \sTwist \sTwist \phi={}&\tfrac{1}{2}(\sTwist \Psi){\overset{1,0}{\odot}}\phi
 -  \tfrac{1}{10}(\sCurlDagger \Psi){\overset{0,0}{\odot}}\phi
 + \tfrac{5}{2}\Psi {\overset{1,0}{\odot}}\sTwist \phi
 + \tfrac{1}{4} m\Psi {\overset{0,0}{\odot}}\chi
 + \Phi {\overset{0,1}{\odot}}\sTwist \phi  \label{eq:order3ReductionToTwists3},\\
\sCurl \sTwist \sTwist \chi={}&- \tfrac{1}{3} m\sTwist \sTwist \phi
 + \tfrac{2}{3}(\sTwist \Phi){\overset{0,1}{\odot}}\chi
 -  \tfrac{2}{9}(\sCurl \Phi){\overset{0,0}{\odot}}\chi
 + \tfrac{8}{3}\Phi {\overset{0,1}{\odot}}\sTwist \chi
 -  \tfrac{1}{3} m\Phi {\overset{0,0}{\odot}}\phi\nonumber\\
& + \Psi {\overset{1,0}{\odot}}\sTwist \chi  \label{eq:order3ReductionToTwists4},\\
\sCurlDagger \sTwist \sTwist \phi={}&\tfrac{1}{3} m\sTwist \sTwist \chi
 + \tfrac{2}{3}(\sTwist \Phi){\overset{1,0}{\odot}}\phi
 -  \tfrac{2}{9}(\sCurlDagger \Phi){\overset{0,0}{\odot}}\phi
 + \tfrac{8}{3}\Phi {\overset{1,0}{\odot}}\sTwist \phi
 + \tfrac{1}{3} m\Phi {\overset{0,0}{\odot}}\chi\nonumber\\
& + \bar\Psi{\overset{0,1}{\odot}}\sTwist \phi  \label{eq:order3ReductionToTwists5},\\
\sCurlDagger \sTwist \sTwist \chi={}&\tfrac{1}{2}(\sTwist \bar\Psi){\overset{0,1}{\odot}}\chi
 -  \tfrac{1}{10}(\sCurl \bar\Psi){\overset{0,0}{\odot}}\chi
 + \tfrac{5}{2}\bar\Psi{\overset{0,1}{\odot}}\sTwist \chi
 -  \tfrac{1}{4} m\bar\Psi{\overset{0,0}{\odot}}\phi
 + \Phi {\overset{1,0}{\odot}}\sTwist \chi  \label{eq:order3ReductionToTwists6}.
\end{align}
\end{subequations}
}%
These are shown in our \emph{Mathematica} notebook~\cite{myGithubRepo}.

\subsection{Decomposing equations}%
\label{sub:decomposing_equations}

A set $\set{\tensor{(\phi_i)}{_{B \dots Q}^{B' \dots Q'}} \text{, } i = 1, 2, \dots}$ of spinor fields subject to a differential equation is an \emph{exact set of fields}~\cite[section 5.10]{PenRinVol1} if, at each spacetime point $P$,
\begin{enumerate}
	\item the symmetrized derivatives $\tensor{\nabla}{_{(A_1}^{(A_1'}} \tensor{(\phi_i)}{_{B \dots Q)}^{B' \dots Q')}}$, $\tensor{\nabla}{_{(A_2}^{(A_2'}} \tensor{\nabla}{_{A_1}^{A_1'}} \tensor{(\phi_i)}{_{B \dots Q)}^{B' \dots Q')}}$, etc., can take arbitrary values, and%
	\label{item:symmetrized_derivatives_take_arbitrary_values}
	\item the unsymmetrized derivatives are determined by the symmetrized derivatives.%
	\label{item:unsymmetrized_derivatives_are_determined}
\end{enumerate}

The Dirac fields form an exact set of fields.
This is a consequence of
\cref{lemma:differential_operator_expressible_as_twists}.
For this reason, we will encounter equations of the types
{\allowdisplaybreaks
\begin{subequations}
\begin{align}
	\tensor{S}{^{A_1}^\dots^{A_{k+1}}^B_{A_1'}_\dots_{A_k'}} \tensor{(\underbrace{\sTwist \sTwist \dots \sTwist}_{\times k} \phi)}{_{A_1}_\dots_{A_{k+1}}^{A_1'}^\dots^{A_k'}} &= 0,%
	\label{eq:term_in_left_dirac}\\
	\tensor{S}{^{A_1}^\dots^{A_{k+1}}_{A_1'}_\dots_{A_k'}^{B'}} \tensor{(\underbrace{\sTwist \sTwist \dots \sTwist}_{\times k} \phi)}{_{A_1}_\dots_{A_{k+1}}^{A_1'}^\dots^{A_k'}} &= 0,\\
	\tensor{S}{^{A_1}^\dots^{A_k}^B_{A_1'}_\dots_{A_{k+1}'}} \tensor{(\underbrace{\sTwist \sTwist \dots \sTwist}_{\times k} \chi)}{_{A_1}_\dots_{A_k}^{A_1'}^\dots^{A_{k+1}'}} &= 0,\\
	\tensor{S}{^{A_1}^\dots^{A_k}_{A_1'}_\dots_{A_{k+1}'}^{B'}} \tensor{(\underbrace{\sTwist \sTwist \dots \sTwist}_{\times k} \chi)}{_{A_1}_\dots_{A_k}^{A_1'}^\dots^{A_{k+1}'}} &= 0,
\end{align}
\end{subequations}
}%
where $\phi_A$ and $\chi_{A'}$ are the Dirac fields and $S$ is a spinor field.
$S$ may without loss of generality be taken to be symmetric in the indices that are contracted since they are contracted with a symmetric spinor.

By \cref{item:symmetrized_derivatives_take_arbitrary_values}, the twists can take arbitrary values at $P$.
Contracting, for example, \cref{eq:term_in_left_dirac} with a test field $T_B$ yields the scalar equation
\begin{align}\label{eq:term_in_left_dirac_contracted_with_test_field}
	\tensor{S}{^{A_1}^\dots^{A_{k+1}}^B_{A_1'}_\dots_{A_k'}} \tensor{(\underbrace{\sTwist \sTwist \dots \sTwist}_{\times k} \phi)}{_{A_1}_\dots_{A_{k+1}}^{A_1'}^\dots^{A_k'}} T_B &= 0.
\end{align}
But since the test field also may take arbitrary values, spinors of the form
\begin{align}
	\tensor{W}{_{A_1}_\dots_{A_{k+1}}^{A_1'}^\dots^{A_k'}_B} \definedas \tensor{(\underbrace{\sTwist \sTwist \dots \sTwist}_{\times k} \phi)}{_{A_1}_\dots_{A_{k+1}}^{A_1'}^\dots^{A_k'}} T_B
\end{align}
span $\tensor{\complexes}{_{(A_1}_\dots_{A_{k+1})}_B^{(A_1'}^\dots^{A_k')}}$.
By \cref{thm:spinor_decomposition}, $\tensor{W}{_{A_1}_\dots_{A_{k+1}}^{A_1'}^\dots^{A_k'}_{B}}$ has two independent parts: $\tensor{W}{_{(A_1}_\dots_{A_{k+1}}^{(A_1'}^\dots^{A_k')}_{B)}}$ and $\tensor{W}{_{(A_1}_\dots_{A_k}^{C}^{A_1'}^\dots^{A_k'}_{|C|}} \tensor{\epsilon}{_{A_{k+1})}_B}$.
Hence \cref{eq:term_in_left_dirac_contracted_with_test_field} splits into
\begin{align}
	0 ={}& \tensor{S}{^{(A_1}^\dots^{A_{k+1}}^{B)}_{(A_1'}_\dots_{A_k')}} \tensor{W}{_{(A_1}_\dots_{A_{k+1}}^{(A_1'}^\dots^{A_k')}_{B)}} \nonumber\\*
	-{}& \frac{k}{k+1} \tensor{S}{^{(A_1 \dots A_k) B}_{B (A_1' \dots A_k')}} \tensor{W}{_{A_1}_\dots_{A_k}^{C}^{A_1'}^\dots^{A_k'}_{C}}.
\end{align}
The two independent parts of $\tensor{W}{_{A_1}_\dots_{A_{k+1}}^{A_1'}^\dots^{A_k'}_{B}}$ may take arbitrary and independent values, so
\begin{subequations}
\begin{align}
	\tensor{S}{^{(A_1}^\dots^{A_{k+1}}^{B)}_{A_1' \dots A_k'}} &= 0,\\
	\tensor{S}{^{A_1 \dots A_k B}_{B A_1' \dots A_k'}} &= 0.
\end{align}
\end{subequations}

This technique is used abundantly when analyzing the equations for the symmetry operators.

\section{Conditions for and form of the symmetry operators}

There is a general method that we can follow to derive conditions for the existence of an $n$:th order symmetry operator $L$.
Firstly, we make an ansatz for $L$ and substitute with this in the Dirac equation.
Secondly, we rewrite the equations to only contain twists using \cref{lemma:differential_operator_expressible_as_twists}.
We then decompose the resulting equations into irreducible parts as in \cref{sub:decomposing_equations} and lastly we simplify.

In this section, we first demonstrate this method by applying it to the zeroth order symmetry operator.
Then the results for the first and second order symmetry operators are stated directly and interpreted.

The main results are
\cref{thm:condition_for_0th_order_symmetry_operator_for_dirac_equation,%
thm:condition_for_1st_order_symmetry_operator_for_dirac_equation,%
thm:condition_for_2nd_order_symmetry_operator_for_dirac_equation}.

\subsection{Zeroth order symmetry operator}%
\label{sub:zeroth_order_symmetry_operator}

Let $\symmetryOperator \from (\tensor{\phi}{_A}, \tensor{\chi}{_{A'}}) \mapsto (\tensor{\lambda}{_A}, \tensor{\gamma}{_{A'}})$ be of the form
\begin{subequations}
\begin{align}
	\tensor{\lambda}{_A} ={}& \tensor{K}{_A^B} \tensor{\phi}{_B} + \tensor{L}{_A^{A'}} \tensor{\chi}{_{A'}},%
	\label{eq:0th_order_symop_left_ansatz}\\
	\tensor{\gamma}{_{A'}} ={}& \tensor{M}{^A_{A'}} \tensor{\phi}{_A} + \tensor{N}{_{A'}^{B'}} \tensor{\chi}{_{B'}}.%
	\label{eq:0th_order_symop_right_ansatz}
\end{align}
\end{subequations}
$\tensor{L}{_A^{A'}}$ and $\tensor{M}{_A^{A'}}$ are already irreducible, but
\begin{subequations}
\begin{align}
\tensor{K}{_{A B}} ={}& -\frac{1}{2} \tensor{K}{_C^C} \tensor{\epsilon}{_{A B}} 
 + \tensor{K}{_{(A B)}},\\
\tensor{N}{^{A' B'}} ={}& -\frac{1}{2} \tensor{N}{^{C'}_{C'}} \tensor{\conjugate\epsilon}{^{A' B'}} 
 + \tensor{N}{^{(A' B')}}, 
\end{align}
\end{subequations}
so we will name these irreducible parts
\begin{align*}
\coeff{K}{0}{0}  ={}& \tensor{K}{^A_A}, \quad& \tensor{\coeff{K}{2}{0}}{_{AB}} ={}& \tensor{K}{_{(A B)}},\\
\coeff{N}{0}{0} = {}& \tensor{N}{^{A'}_{A'}}, \quad& \tensor{\coeff{N}{0}{2}}{^{A' B'}} ={}& \tensor{N}{^{(A' B')}},
\end{align*}
where the underscript indicates the valence numbers for totally symmetric spinors.
Substituting this into \cref{eq:left_dirac_for_lambda,eq:right_dirac_for_gamma}, we have that
\begin{subequations}
\begin{align}
	\sCurlDagger \left( -L{\overset{0,1}{\odot }}\chi + \frac{1}{2} \underset{0,0}{K}{}{\overset{0,0}{\odot }}\phi + \underset{2,0}{K}{} {\overset{1,0}{\odot }}\phi \right) ={}&
	\frac{1}{2} m\underset{0,0}{N}{} {\overset{0,0}{\odot }}\chi - m \underset{0,2}{N}{}{\overset{0,1}{\odot }}\chi - m M{\overset{1,0}{\odot }}\phi,\\
	\sCurl \left( \frac{1}{2} \underset{0,0}{N}{}{\overset{0,0}{\odot }}\chi - \underset{0,2}{N}{}{\overset{0,1}{\odot }}\chi - M {\overset{1,0}{\odot }}\phi \right) ={}&
	m L{\overset{0,1}{\odot }}\chi - \frac{1}{2} m\underset{0,0}{K}{}{\overset{0,0}{\odot }}\phi + m \underset{2,0}{K}{}{\overset{1,0}{\odot }}\phi.
\end{align}
\end{subequations}
Applying the Leibniz rules from \cref{lemma:symmult_liebniz_rules} yields
\begin{subequations}
\begin{align}
	0 ={}& - L{\overset{1,1}{\odot}} \sTwist \chi + \frac{1}{2} L {\overset{1,0}{\odot}}\sCurl \chi  + (\sCurlDagger L) {\overset{0,1}{\odot}}\chi - \frac{1}{2} (\sDiv L) {\overset{0,0}{\odot}}\chi - \frac{1}{2} (\sTwist \underset{0,0}{K}{}) {\overset{1,0}{\odot}}\phi - \frac{1}{2}\underset{0,0}{K}{}{\overset{0,0}{\odot}}\sCurlDagger \phi \nonumber\\
	&{} - \underset{2,0}{K}{}{\overset{2,0}{\odot}}\sTwist \phi + (\sCurlDagger \underset{2,0}{K}{}) {\overset{1,0}{\odot}} \phi + \frac{1}{2} m\underset{0,0}{N}{}{\overset{0,0}{\odot}}\chi - m \underset{0,2}{N}{} {\overset{0,1}{\odot}}\chi - m M {\overset{1,0}{\odot}}\phi,\\
	0 ={}& -\frac{1}{2} (\sTwist \underset{0,0}{N}{}) {\overset{0,1}{\odot}}\chi - \frac{1}{2}\underset{0,0}{N}{}{\overset{0,0}{\odot}}\sCurl \chi - \underset{0,2}{N}{}{\overset{0,2}{\odot}}\sTwist \chi + (\sCurl \underset{0,2}{N}{}) {\overset{0,1}{\odot}}\chi - M {\overset{1,1}{\odot}}\sTwist \phi - \frac{1}{2}M {\overset{0,1}{\odot}}\sCurlDagger \phi \nonumber\\
	&{} (\sCurl M) {\overset{1,0}{\odot}}\phi - \frac{1}{2} (\sDiv M) {\overset{0,0}{\odot}} \phi  + m L {\overset{0,1}{\odot}}\chi - \frac{1}{2} m\underset{0,0}{K}{}{\overset{0,0}{\odot}}\phi + m \underset{2,0}{K}{} {\overset{1,0}{\odot}}\phi.
\end{align}
\end{subequations}
Using \cref{lemma:differential_operator_expressible_as_twists}, this can be rewritten in terms of only twists:
\begin{subequations}
\begin{align}
	0={}&- L {\overset{1,1}{\odot}}\sTwist \chi
 	+ \frac{1}{2} mL {\overset{1,0}{\odot}}\phi
 	+L {\overset{0,1}{\odot}}\sCurlDagger\chi  
 	- \frac{1}{2} ( \sDiv L) {\overset{0,0}{\odot}}\chi
 	- \frac{1}{2}( \sTwist \underset{0,0}{K}{}) {\overset{1,0}{\odot}}\phi
 	- \frac{1}{2} m\underset{0,0}{K}{}{\overset{0,0}{\odot}}\chi\nonumber\\
	& - \underset{2,0}{K}{}{\overset{2,0}{\odot}}\sTwist \phi
 	+ (\sCurlDagger \underset{2,0}{K}{}) {\overset{1,0}{\odot}}\phi 
 	+ \frac{1}{2} m\underset{0,0}{N}{}{\overset{0,0}{\odot}}\chi
 	+ m (\underset{0,2}{N}{}) {\overset{0,1}{\odot}}\chi
 	- m M {\overset{1,0}{\odot}}\phi ,%
 	\label{eq:left_dirac_for_lambda_in_terms_of_twists}\\
	0={}&-\frac{1}{2}(\sTwist \underset{0,0}{N}{}) {\overset{0,1}{\odot}}\chi
 	+ \frac{1}{2} m\underset{0,0}{N}{}{\overset{0,0}{\odot}}\phi
 	- \underset{0,2}{N}{}{\overset{0,2}{\odot}}\sTwist \chi
 	+ (\sCurl \underset{0,2}{N}{}) {\overset{0,1}{\odot}}\chi
 	- M{\overset{1,1}{\odot}}\sTwist \phi
 	- \frac{1}{2} mM {\overset{0,1}{\odot}}\chi\nonumber\\
	& + (\sCurl M) {\overset{1,0}{\odot}}\phi
	- \frac{1}{2}(\sDiv M) {\overset{0,0}{\odot}}\phi
	+ m L{\overset{0,1}{\odot}}\chi 
	- \frac{1}{2} m\underset{0,0}{K}{}{\overset{0,0}{\odot}}\phi
	+  m\underset{2,0}{K}{}{\overset{1,0}{\odot}}\phi.%
	\label{eq:right_dirac_for_gamma_in_terms_of_twists}
\end{align}
\end{subequations}

Now, since each order of derivative is independent by \cref{sub:decomposing_equations}, and since each field is independent and free, \cref{eq:left_dirac_for_lambda_in_terms_of_twists,eq:right_dirac_for_gamma_in_terms_of_twists} splits into eight equations.

\subsubsection{Collecting first order terms}%
\label{sub:first_order_terms}

Isolating the $\sTwist \phi$-terms of \cref{eq:left_dirac_for_lambda_in_terms_of_twists} yields
\begin{subequations}
\begin{align}\label{eq:twist_phi_condition1_for_0th_order_symop}
	0={}& \coeff{K}{2}{0} {\overset{2,0}{\odot}} \sTwist \phi.
\end{align}
Isolating the $\sTwist \chi$-terms of \cref{eq:left_dirac_for_lambda_in_terms_of_twists} yields
\begin{align}\label{eq:twist_chi_condition1_for_0th_order_symop}
 	0={} L {\overset{1, 1}{\odot}} \sTwist \chi.
\end{align}
Isolating the $\sTwist \phi$-terms of \cref{eq:right_dirac_for_gamma_in_terms_of_twists} yields
\begin{align}\label{eq:twist_phi_condition2_for_0th_order_symop}
 	0={} M {\overset{1,1}{\odot}} \sTwist \phi.
\end{align}
Isolating the $\sTwist \chi$-terms of \cref{eq:right_dirac_for_gamma_in_terms_of_twists} yields
\begin{align}\label{eq:twist_chi_condition2_for_0th_order_symop}
 	0={} \coeff{N}{0}{2} {\overset{0,2}{\odot}} \sTwist \chi.
\end{align}
\end{subequations}

The reason for introducing $\coeff{K}{2}{0}$ and $\coeff{N}{0}{2}$ is that
\cref{eq:twist_phi_condition1_for_0th_order_symop,%
eq:twist_chi_condition1_for_0th_order_symop,%
eq:twist_phi_condition2_for_0th_order_symop,%
eq:twist_chi_condition2_for_0th_order_symop}
are irreducible in the sense of \cref{sub:decomposing_equations}.
It follows that
\begin{subequations}
\begin{align}
	\underset{2,0}{K}{}={}&0 \label{eq:order0SymOpOrder1Conditions1},\\
	L={}&0 \label{eq:order0SymOpOrder1Conditions2},\\
	M={}&0 \label{eq:order0SymOpOrder1Conditions3},\\
	\underset{0,2}{N}{}={}&0 \label{eq:order0SymOpOrder1Conditions4}.
\end{align}
\end{subequations}

\subsubsection{Collecting zeroth order terms}%
\label{sub:zeroth_order_terms}

Isolating the $\phi$-terms of \cref{eq:left_dirac_for_lambda_in_terms_of_twists} yields
\begin{subequations}
\begin{align}\label{eq:phi_condition1_for_0th_order_symop}
	0={}&
	\frac{1}{2} m L {\overset{1, 0}{\odot}} \phi
	- \frac{1}{2} (\sTwist \coeff{K}{0}{0}) {\overset{1, 0}{\odot}} \phi
	+ \underset{2,0}{K}{} {\overset{1,0}{\odot}}\sCurlDagger \phi
	- m M {\overset{1,0}{\odot}} \phi.
\end{align}
Isolating the $\chi$-terms of \cref{eq:left_dirac_for_lambda_in_terms_of_twists} yields
\begin{align}\label{eq:chi_condition1_for_0th_order_symop}
	0={}&
 	(\sCurlDagger L) {\overset{0,1}{\odot}} \chi
 	+ \frac{1}{2} (\sDiv L) {\overset{0,0}{\odot}} \chi
 	- \frac{1}{2} m \underset{0,0}{K}{}{\overset{0,0}{\odot}} \chi
 	+ \frac{1}{2} m \underset{0,0}{N}{}{\overset{0,0}{\odot}} \chi
 	- m \underset{0,2}{N}{} {\overset{0,1}{\odot}} \chi.
\end{align}
Isolating the $\phi$-terms of \cref{eq:right_dirac_for_gamma_in_terms_of_twists} yields
\begin{align}\label{eq:phi_condition2_for_0th_order_symop}
	0={}&
 	\frac{1}{2} m\underset{0,0}{N}{}{\overset{0,0}{\odot}}\phi
	+ (\sCurl M) {\overset{1,0}{\odot}} \phi
	- \frac{1}{2} (\sDiv M) {\overset{0,0}{\odot}} \phi
	- \frac{1}{2} m \underset{0,0}{K}{} {\overset{0,0}{\odot}} \phi
	+ m \underset{2,0}{K}{} {\overset{1,0}{\odot}} \phi.
\end{align}
Isolating the $\chi$-terms of \cref{eq:right_dirac_for_gamma_in_terms_of_twists} yields
\begin{align}\label{eq:chi_condition2_for_0th_order_symop}
	0={}&
	-\frac{1}{2} (\sTwist \underset{0,0}{N}{}) {\overset{0,1}{\odot}} \chi
 	+ (\sCurl \underset{0,2}{N}{}) {\overset{0,1}{\odot}} \chi
 	- \frac{1}{2} m M {\overset{0,1}{\odot}} \chi
	- m L {\overset{0,1}{\odot}} \chi.
\end{align}
\end{subequations}

Using
\cref{eq:order0SymOpOrder1Conditions1,%
eq:order0SymOpOrder1Conditions2,%
eq:order0SymOpOrder1Conditions3,%
eq:order0SymOpOrder1Conditions4},
\cref{eq:phi_condition1_for_0th_order_symop,%
eq:chi_condition1_for_0th_order_symop,%
eq:phi_condition2_for_0th_order_symop,%
eq:chi_condition2_for_0th_order_symop}
reduce to
\begin{subequations}
\begin{align}
	\sTwist \underset{0,0}{K}{}={}&0 \label{eq:order0SymOpOrder0Conditions1},\\
	-\frac{1}{2} m\underset{0,0}{K}{} + \frac{1}{2} m\underset{0,0}{N}{}={}&0 \label{eq:order0SymOpOrder0Conditions2},\\
	\sTwist \underset{0,0}{N}{}={}&0 \label{eq:order0SymOpOrder0Conditions3}.
\end{align}
\end{subequations}

\subsubsection{Interpretation and discussion}%
\label{sub:zeroth_order_symmetry_operator:interpretation_and_discussion}

To interpret these equations, note that \cref{eq:order0SymOpOrder1Conditions2,eq:order0SymOpOrder1Conditions3} imply that there is no mixing between $\phi_A$ and $\chi_{A'}$.
\cref{eq:order0SymOpOrder1Conditions1,eq:order0SymOpOrder1Conditions4} imply that the only non-zero parts of $\tensor{K}{_A_B}$ and $\tensor{N}{^{A'}^{B'}}$ are the trace parts. That is, they are proportional to the identity.
The twists in \cref{eq:order0SymOpOrder0Conditions1,eq:order0SymOpOrder0Conditions3} act on valence $(0, 0)$ spinors, so they are just covariant derivatives.
Hence $\coeff{K}{0}{0}$ and $\coeff{N}{0}{0}$ must be constant.
Since we assume that $m \neq 0$, we may divide by it in \cref{eq:order0SymOpOrder0Conditions1} and deduce that $\coeff{K}{0}{0}$ and $\coeff{N}{0}{0}$ are equal.

Substituting this into the ansatz, \cref{eq:0th_order_symop_left_ansatz,eq:0th_order_symop_right_ansatz}, that we made for $L$, we get that
\begin{theorem}\label{thm:condition_for_0th_order_symmetry_operator_for_dirac_equation}
	The only zeroth order symmetry operators for the Dirac equation are multiples of the identity.
\end{theorem}

\bigskip

In the following sections, the same general method is scaled up by the use of computer algebra and applied to first and second order symmetry operators.

\subsection{First order symmetry operator}%
\label{sub:first_order_symmetry_operator}

Let $\symmetryOperator \from (\phi_A, \chi_{A'}) \mapsto (\lambda_A, \gamma_{A'})$ be of the form
\begin{subequations}
\begin{align}
\lambda_{A}={}&
	K1_{A}{}^{BCA'} (\sTwist \phi)_{BCA'}
	+ L1_{A}{}^{BA'B'} (\sTwist \chi)_{BA'B'}
	+ K0_{A}{}^{B} \phi_{B}
	+ L0_{A}{}^{A'} \chi_{A'}\,,\\
\gamma_{A'}={}&
	M1^{AB}{}_{A'}{}^{B'} (\sTwist \phi)_{ABB'}
	+ N1^{A}{}_{A'}{}^{B'C'} (\sTwist \chi)_{AB'C'}
	+ M0^{A}{}_{A'} \phi_{A}
	+ N0_{A'}{}^{B'}\chi_{B'}\,.
\end{align}
\end{subequations}
By \cref{lemma:differential_operator_expressible_as_twists}, this is the most general form of a first order symmetry operator.

As before, we substitute this into \cref{eq:left_dirac_for_lambda,eq:right_dirac_for_gamma}, collect each order of derivative, and decompose the resulting equations.
There are then in total 18 equations and 12 variables.
They are not stated here since they are terribly complicated while adding nothing conceptually different from \cref{sub:zeroth_order_symmetry_operator}.
The calculations are, in their entirety, available on Github~\cite{myGithubRepo}. 
After simplification, they may be expressed as \cref{thm:condition_for_1st_order_symmetry_operator_for_dirac_equation}.

\begin{definition}\label{def:auxiliary_condition_A}
	Let $\tensor{S}{_A^{A'}}$, $\tensor{T}{^{A'}^{B'}}$, $\tensor{U}{_A_B}$, and $\tensor{R}{_A^{A'}}$ be Killing spinors on a Lorentzian manifold $M$.
	They satisfy \emph{auxiliary condition A} if
	\begin{subequations}
	\begin{align}
		\sCurl S={}&0 \label{eq:order1AuxiliaryConditions1},\\
		\sCurlDagger S={}&0 \label{eq:order1AuxiliaryConditions2},\\
		\sDiv R={}&0 \label{eq:order1AuxiliaryConditions3},\\
		\sCurl T + \sCurlDagger U={}&0 \label{eq:order1AuxiliaryConditions4}.
	\end{align}
	\end{subequations}
\end{definition}

\begin{theorem}\label{thm:condition_for_1st_order_symmetry_operator_for_dirac_equation}
	The massive Dirac equation has a first order symmetry operator if and only if there exist Killing spinors (not all zero) satisfying auxiliary condition A.
	The symmetry operator then takes the form
	\begin{subequations}
	\begin{align}
\lambda ={}&R{\overset{1,1}{\odot}}\sTwist \phi
 + S{\overset{1,1}{\odot}}\sTwist \phi
 + \bigl(O + \tfrac{3}{8} (\sDiv S)\bigr)\phi
 -  \tfrac{1}{2}(\sCurl R){\overset{1,0}{\odot}}\phi
 -  mU{\overset{1,0}{\odot}}\phi
 + T{\overset{0,2}{\odot}}\sTwist \chi\nonumber\\
& + \tfrac{2}{3}(\sCurlDagger U){\overset{0,1}{\odot}}\chi
 -  \tfrac{1}{2} mR{\overset{0,1}{\odot}}\chi
 + \tfrac{3}{2} mS{\overset{0,1}{\odot}}\chi  \label{eq:formOf1stOrderSymOpLambda},\\
\gamma ={}&U{\overset{2,0}{\odot}}\sTwist \phi
 -  \tfrac{2}{3}(\sCurlDagger U){\overset{1,0}{\odot}}\phi
 + \tfrac{1}{2} mR{\overset{1,0}{\odot}}\phi
 + \tfrac{3}{2} mS{\overset{1,0}{\odot}}\phi
 -  S{\overset{1,1}{\odot}}\sTwist \chi
 + R{\overset{1,1}{\odot}}\sTwist \chi\nonumber\\
& + \bigl(O -  \tfrac{3}{8} (\sDiv S)\bigr)\chi
 -  \tfrac{1}{2}(\sCurlDagger R){\overset{0,1}{\odot}}\chi
 + mT{\overset{0,1}{\odot}}\chi  \label{eq:formOf1stOrderSymOpGamma},
\end{align}
	\end{subequations}
	for some constant scalar $O$.
\end{theorem}

\subsubsection{Interpretation and discussion}%
\label{sub:first_order_symmetry_operator:interpretation_and_discussion}

The geometric interpretation of \cref{eq:order1AuxiliaryConditions1,eq:order1AuxiliaryConditions2} is that $S_A{}^{A'}$ is a closed vector field.
The geometric interpretation of \cref{eq:order1AuxiliaryConditions3} is that $\tensor{R}{_A^{A'}}$ is a Killing vector.
Observe that if $\Phi=0$, then $\sCurl T$ is a Killing vector because $\sTwist\sCurl T=0$ and $\sDiv\sCurl T=0$ due to \eqref{eq:curl_twist_to_twist_curl} and \eqref{eq:div_curl_to_curl_div}.
Similarly, $\sCurlDagger U$ is then also a Killing vector.

The different possible algebraic types of the Weyl spinor are commonly classified by \emph{Petrov type}.
The existence of a nontrivial valence $(2, 0)$ spinor implies that the spacetime is of type D, N, or O~\cite[section 4.7]{AndBackBlue18}.
The geometric interpretation of \cref{eq:order1AuxiliaryConditions4} is that $\tensor{f}{_A^{A'}_B^{B'}} \definedas \tensor{U}{_A_B} \tensor{\conjugate \epsilon}{^{A'}^{B'}} + \tensor{\epsilon}{_A_B} \tensor{T}{^{A'}^{B'}}$ is a Killing--Yano tensor.
This is shown in our \emph{Mathematica} notebook~\cite{myGithubRepo}.

Kamran and McLenaghan~\cite[theorem II]{KamMcL84} have derived the form of the most general first order symmetry operator for the massive Dirac equation using the Dirac basis.
\Cref{thm:condition_for_1st_order_symmetry_operator_for_dirac_equation} is a covariant reformulation of their result.

Lastly for this subsection, let's look at \cref{eq:formOf1stOrderSymOpLambda,eq:formOf1stOrderSymOpGamma} in terms of Lie derivatives.
In \cref{sub:killing_tensors}, we stated that Killing vectors generate infinitesimal isometries, so one might expect that taking a Lie derivative with respect to $R$ is a symmetry operation.
This is true if one takes the conformally weighted Lie derivative.
\cref{eq:formOf1stOrderSymOpLambda,eq:formOf1stOrderSymOpGamma} may be written
\begin{subequations}
\begin{align}
 \lambda ={}&
 \hat{\mathcal{L}}_{R}\phi
 + \hat{\mathcal{L}}_{S}\phi
 + O\phi
 -  mU{\overset{1,0}{\odot}}\phi
 + T{\overset{0,2}{\odot}}\sTwist \chi
 + \tfrac{2}{3}(\sCurlDagger U){\overset{0,1}{\odot}}\chi
 + 2 mS{\overset{0,1}{\odot}}\chi ,\\
\gamma ={}&
 U{\overset{2,0}{\odot}}\sTwist \phi
 - \tfrac{2}{3}(\sCurlDagger U){\overset{1,0}{\odot}}\phi
 + 2 mS{\overset{1,0}{\odot}}\phi
 + \hat{\mathcal{L}}_{R}\chi
 -  \hat{\mathcal{L}}_{S}\chi
 + O\chi
 + mT{\overset{0,1}{\odot}}\chi .
\end{align}
\end{subequations}
This is shown in \emph{Mathematica}~\cite{myGithubRepo}.

\subsection{Second order symmetry operator}%
\label{sub:second_order_symmetry_operator}

Let $\symmetryOperator \from (\phi_A, \chi_{A'}) \mapsto (\lambda_A, \gamma_{A'})$ be of the form
\begin{subequations}
\begin{align}
	\lambda_{A}={}
	& K2_{A}{}^{BCDA'B'} (\sTwist \sTwist \phi)_{BCDA'B'} + L2_{A}{}^{BCA'B'C'} (\sTwist \sTwist \chi)_{BCA'B'C'}\nonumber\\
	&+ K1_{A}{}^{BCA'} (\sTwist \phi)_{BCA'} + L1_{A}{}^{BA'B'} (\sTwist \chi)_{BA'B'}\nonumber\\
	&+ K0_{A}{}^{B} \phi_{B} + L0_{A}{}^{A'} \chi_{A'} \label{eq:2nd_order_ansatz_lambda},\\
	\gamma_{A'}={}
	& N2^{AB}{}_{A'}{}^{B'C'D'} (\sTwist \sTwist \chi)_{ABB'C'D'} + M2^{ABC}{}_{A'}{}^{B'C'} (\sTwist \sTwist \phi)_{ABCB'C'}\nonumber\\
	&+ N1^{A}{}_{A'}{}^{B'C'} (\sTwist \chi)_{AB'C'} + M1^{AB}{}_{A'}{}^{B'} (\sTwist \phi)_{ABB'}\nonumber\\
	&+ M0^{A}{}_{A'} \phi_{A} + \chi^{B'} N0_{A'B'} \label{eq:2nd_order_ansatz_gamma}.
\end{align}
\end{subequations}

As before, we substitute this into \cref{eq:left_dirac_for_lambda,eq:right_dirac_for_gamma}, collect each order of derivative, and decompose the resulting equations.
There are then in total 26 equations and 20 variables.
Simplifying those gives us \cref{thm:condition_for_2nd_order_symmetry_operator_for_dirac_equation}.

\begin{definition}\label{def:auxiliary_condition_B}
	Let $\tensor{V}{_A_B^{A'}^{B'}}$, $\tensor{W}{_A_B^{A'}^{B'}}$, $\tensor{X}{_A_B_C^{A'}}$, and $\tensor{Y}{_A^{A'}^{B'}^{C'}}$ be Killing spinors on a Lorentzian manifold.
	They satisfy \emph{auxiliary condition B} if there exist spinors $\tensor{R}{_A^{A'}}$, $\tensor{T}{^{A'}^{B'}}$, $\tensor{U}{_A_B}$, $\tensor{S}{_A^{A'}}$, and a scalar $O$ such that 
	{\allowdisplaybreaks
	\begin{subequations}
 \begin{align}
\sTwist S={}&\tfrac{1}{3}\bar\Psi{\overset{0,2}{\odot}}V
 -  \tfrac{1}{3}\Psi {\overset{2,0}{\odot}}V
 + \tfrac{1}{4} m\sCurl Y
 + \tfrac{1}{4} m\sCurlDagger X \label{eq:order2AuxiliaryConditions6},\\
\sCurl S={}&{-} \frac{2}{5 m}X{\overset{3,1}{\odot}}\sTwist \Psi
 + \frac{9}{25 m}X{\overset{2,1}{\odot}}\sCurl \Phi
 -  \frac{9}{40 m}\Psi {\overset{2,0}{\odot}}\sDiv X
 -  \frac{1}{2 m}\Psi {\overset{3,0}{\odot}}\sCurl X\nonumber\\*
& -  \tfrac{1}{2} m\sDiv X
 + \tfrac{1}{3}\sDiv \sCurl V
 -  \tfrac{4}{3}\Phi {\overset{1,2}{\odot}}V \label{eq:order2AuxiliaryConditions7},\\
\sCurlDagger S={}&{-} \frac{1}{2 m}\bar\Psi{\overset{0,3}{\odot}}\sCurlDagger Y
 + \frac{9}{25 m}Y{\overset{1,2}{\odot}}\sCurlDagger \Phi
 -  \frac{2}{5 m}Y{\overset{1,3}{\odot}}\sTwist \bar\Psi
 -  \frac{9}{40 m}\bar\Psi{\overset{0,2}{\odot}}\sDiv Y\nonumber\\*
& -  \tfrac{1}{2} m\sDiv Y
 -  \tfrac{1}{3}\sDiv \sCurlDagger V
 + \tfrac{4}{3}\Phi {\overset{2,1}{\odot}}V \label{eq:order2AuxiliaryConditions8},\\
\sTwist O={}&\tfrac{3}{10}V{\overset{2,1}{\odot}}\sCurl \Phi
 + \tfrac{3}{10}V{\overset{1,2}{\odot}}\sCurlDagger \Phi
 -  \tfrac{2}{3} m^2\sDiv V
 + \tfrac{1}{8} m\sDiv \sCurl Y
 -  \tfrac{1}{2} m\bar\Psi{\overset{0,3}{\odot}}Y\nonumber\\*
& -  \tfrac{1}{2} m\Phi {\overset{1,2}{\odot}}Y
 + \tfrac{1}{2} m\Phi {\overset{2,1}{\odot}}X
 -  \tfrac{1}{10}\Psi {\overset{3,0}{\odot}}\sCurl V
 -  \tfrac{1}{10}\bar\Psi{\overset{0,3}{\odot}}\sCurlDagger V
 + \tfrac{1}{2} m\Psi {\overset{3,0}{\odot}}X\nonumber\\*
& -  \tfrac{1}{8} m\sDiv \sCurlDagger X \label{eq:order2AuxiliaryConditions9},
\end{align}
	\end{subequations}
	\begin{subequations}
\begin{align}
\sTwist R={}&\tfrac{1}{3}\bar\Psi{\overset{0,2}{\odot}}W
 -  \tfrac{1}{3}\Psi {\overset{2,0}{\odot}}W \label{eq:order2AuxiliaryConditions1},\\
\sDiv R={}&0 \label{eq:order2AuxiliaryConditions2},\\
\sTwist T={}& {-}\tfrac{2}{3} m\sCurlDagger W \label{eq:order2AuxiliaryConditions3},\\
\sTwist U={}& {-}\tfrac{2}{3} m\sCurl W \label{eq:order2AuxiliaryConditions4},\\
\sCurl T + \sCurlDagger U={}&{-} \frac{1}{2 m}W{\overset{2,1}{\odot}}\sCurl \Phi
 + \frac{3}{5 m}W{\overset{2,2}{\odot}}\sTwist \Phi
 -  \frac{1}{2 m}W{\overset{1,2}{\odot}}\sCurlDagger \Phi
 -  \frac{4}{5 m}W{\overset{1,1}{\odot}}\sTwist \Lambda\nonumber\\*
& + \frac{2}{5 m}\Phi {\overset{2,1}{\odot}}\sCurl W
 + \frac{2}{5 m}\Phi {\overset{1,2}{\odot}}\sCurlDagger W
 + \frac{4}{15 m}\Phi {\overset{1,1}{\odot}}\sDiv W
 -  \frac{1}{15 m}\sTwist \sDiv \sDiv W\nonumber\\*
& + \tfrac{4}{3} m\sDiv W
 + \frac{3}{10 m}\Psi {\overset{3,0}{\odot}}\sCurl W
 + \frac{3}{10 m}\bar\Psi{\overset{0,3}{\odot}}\sCurlDagger W \label{eq:order2AuxiliaryConditions5}.
 \end{align}
	\end{subequations}
	}
\end{definition}

\begin{theorem}\label{thm:condition_for_2nd_order_symmetry_operator_for_dirac_equation}
	The massive Dirac equation has a second order symmetry operator if and only if there exist Killing spinors (not all zero) satisfying auxiliary condition B.
	The symmetry operator is then a linear combination of a symmetry operator of the \emph{first kind}, 
	{\allowdisplaybreaks
	\begin{subequations}
\begin{align}
\lambda ={}&V{\overset{2,2}{\odot}}\sTwist \sTwist \phi
 -  \tfrac{2}{3}(\sCurl V){\overset{2,1}{\odot}}\sTwist \phi
 + \tfrac{8}{9}(\sDiv V){\overset{1,1}{\odot}}\sTwist \phi
 -  mX{\overset{2,1}{\odot}}\sTwist \phi
 + S{\overset{1,1}{\odot}}\sTwist \phi\nonumber\\
& -  \tfrac{1}{3}(\sCurl \sDiv V){\overset{1,0}{\odot}}\phi
 + \tfrac{1}{3}(\Phi {\overset{1,2}{\odot}}V){\overset{1,0}{\odot}}\phi
 -  \frac{9}{50 m}(X{\overset{2,1}{\odot}}\sCurl \Phi){\overset{1,0}{\odot}}\phi
 + \frac{9}{80 m}(\Psi {\overset{2,0}{\odot}}\sDiv X){\overset{1,0}{\odot}}\phi\nonumber\\
& + \frac{1}{5 m}(X{\overset{3,1}{\odot}}\sTwist \Psi){\overset{1,0}{\odot}}\phi
 + \frac{1}{4 m}(\Psi {\overset{3,0}{\odot}}\sCurl X){\overset{1,0}{\odot}}\phi\nonumber\\
& + \bigl(O + \tfrac{3}{8} (\sDiv S) + \tfrac{2}{15} (\sDiv \sDiv V) -  \tfrac{8}{15} (\Phi {\overset{2,2}{\odot}}V)\bigr)\phi
 + Y{\overset{1,3}{\odot}}\sTwist \sTwist \chi
 -  \tfrac{3}{4}(\sCurl Y){\overset{1,2}{\odot}}\sTwist \chi\nonumber\\
& + \tfrac{3}{4}(\sDiv Y){\overset{0,2}{\odot}}\sTwist \chi
 -  \tfrac{2}{3} mV{\overset{1,2}{\odot}}\sTwist \chi
 -  \tfrac{1}{4}(\sCurl \sDiv Y){\overset{0,1}{\odot}}\chi
 + \tfrac{1}{2}(\Phi {\overset{1,2}{\odot}}Y){\overset{0,1}{\odot}}\chi
 + \tfrac{3}{2} mS{\overset{0,1}{\odot}}\chi  \label{eq:formOf2ndOrderSymOpLambda2ndKind},\\
\gamma ={}&X{\overset{3,1}{\odot}}\sTwist \sTwist \phi
 -  \tfrac{3}{4}(\sCurlDagger X){\overset{2,1}{\odot}}\sTwist \phi
 + \tfrac{3}{4}(\sDiv X){\overset{2,0}{\odot}}\sTwist \phi
 + \tfrac{2}{3} mV{\overset{2,1}{\odot}}\sTwist \phi
 -  \tfrac{1}{4}(\sCurlDagger \sDiv X){\overset{1,0}{\odot}}\phi\nonumber\\
& + \tfrac{1}{2}(\Phi {\overset{2,1}{\odot}}X){\overset{1,0}{\odot}}\phi
 + \tfrac{3}{2} mS{\overset{1,0}{\odot}}\phi
 + V{\overset{2,2}{\odot}}\sTwist \sTwist \chi
 -  S{\overset{1,1}{\odot}}\sTwist \chi
 -  \tfrac{2}{3}(\sCurlDagger V){\overset{1,2}{\odot}}\sTwist \chi\nonumber\\
& + \tfrac{8}{9}(\sDiv V){\overset{1,1}{\odot}}\sTwist \chi
 + mY{\overset{1,2}{\odot}}\sTwist \chi
 -  \tfrac{1}{3}(\sCurlDagger \sDiv V){\overset{0,1}{\odot}}\chi
 + \tfrac{1}{3}(\Phi {\overset{2,1}{\odot}}V){\overset{0,1}{\odot}}\chi\nonumber\\
& -  \frac{1}{4 m}(\bar\Psi{\overset{0,3}{\odot}}\sCurlDagger Y){\overset{0,1}{\odot}}\chi
 -  \frac{1}{5 m}(Y{\overset{1,3}{\odot}}\sTwist \bar\Psi){\overset{0,1}{\odot}}\chi
 -  \frac{9}{80 m}(\bar\Psi{\overset{0,2}{\odot}}\sDiv Y){\overset{0,1}{\odot}}\chi\nonumber\\
& + \frac{9}{50 m}(Y{\overset{1,2}{\odot}}\sCurlDagger \Phi){\overset{0,1}{\odot}}\chi
 + \bigl(O -  \tfrac{3}{8} (\sDiv S) + \tfrac{2}{15} (\sDiv \sDiv V) -  \tfrac{8}{15} (\Phi {\overset{2,2}{\odot}}V)\bigr)\chi  \label{eq:formOf2ndOrderSymOpGamma2ndKind},
\end{align}
	\end{subequations}
	and a symmetry operator of the \emph{second kind},
	\begin{subequations}
\begin{align}
\lambda ={}&W{\overset{2,2}{\odot}}\sTwist \sTwist \phi
 -  \tfrac{2}{3}(\sCurl W){\overset{2,1}{\odot}}\sTwist \phi
 + \tfrac{8}{9}(\sDiv W){\overset{1,1}{\odot}}\sTwist \phi
 + R{\overset{1,1}{\odot}}\sTwist \phi
 -  \tfrac{1}{2}(\sCurl R){\overset{1,0}{\odot}}\phi\nonumber\\
& -  \tfrac{2}{9}(\sCurl \sDiv W){\overset{1,0}{\odot}}\phi
 -  mU{\overset{1,0}{\odot}}\phi
 + \bigl(\tfrac{1}{9} (\sDiv \sDiv W) -  \tfrac{1}{3} (\Phi {\overset{2,2}{\odot}}W)\bigr)\phi
 + \tfrac{4}{3} mW{\overset{1,2}{\odot}}\sTwist \chi
 + T{\overset{0,2}{\odot}}\sTwist \chi\nonumber\\
& -  \tfrac{2}{3}(\sCurl T){\overset{0,1}{\odot}}\chi
 -  \tfrac{1}{2} mR{\overset{0,1}{\odot}}\chi
 + \tfrac{4}{9} m(\sDiv W){\overset{0,1}{\odot}}\chi  \label{eq:formOf2ndOrderSymOpLambda1stKind},\\
\gamma ={}&\tfrac{4}{3} mW{\overset{2,1}{\odot}}\sTwist \phi
 + U{\overset{2,0}{\odot}}\sTwist \phi
 -  \tfrac{2}{3}(\sCurlDagger U){\overset{1,0}{\odot}}\phi
 + \tfrac{4}{9} m(\sDiv W){\overset{1,0}{\odot}}\phi
 + \tfrac{1}{2} mR{\overset{1,0}{\odot}}\phi
 -  W{\overset{2,2}{\odot}}\sTwist \sTwist \chi\nonumber\\
& -  \tfrac{8}{9}(\sDiv W){\overset{1,1}{\odot}}\sTwist \chi
 + \tfrac{2}{3}(\sCurlDagger W){\overset{1,2}{\odot}}\sTwist \chi
 + R{\overset{1,1}{\odot}}\sTwist \chi
 -  \tfrac{1}{2}(\sCurlDagger R){\overset{0,1}{\odot}}\chi
 + \tfrac{2}{9}(\sCurlDagger \sDiv W){\overset{0,1}{\odot}}\chi\nonumber\\
& + mT{\overset{0,1}{\odot}}\chi
 + \bigl(- \tfrac{1}{9} (\sDiv \sDiv W) + \tfrac{1}{3} (\Phi {\overset{2,2}{\odot}}W)\bigr)\chi  \label{eq:formOf2ndOrderSymOpGamma1stKind}.
\end{align}
	\end{subequations}
	}
\end{theorem}

\subsubsection{Interpretation and discussion}%

While these equations are much longer and ungainlier than auxiliary condition A, it is worth to note that
\cref{eq:order2AuxiliaryConditions6,%
eq:order2AuxiliaryConditions7,%
eq:order2AuxiliaryConditions8,%
eq:order2AuxiliaryConditions9}
and 
\cref{eq:order2AuxiliaryConditions1,%
eq:order2AuxiliaryConditions2,%
eq:order2AuxiliaryConditions3,%
eq:order2AuxiliaryConditions4,%
eq:order2AuxiliaryConditions5}
are completely decoupled.
They contain different variables from each other.
Hence dividing the symmetry operator into first and second kind.

Also, if $V$, $W$, $X$ and $Y$ are set to zero, we get back auxiliary condition A, since then
\cref{eq:order2AuxiliaryConditions1,%
eq:order2AuxiliaryConditions3,%
eq:order2AuxiliaryConditions4,%
eq:order2AuxiliaryConditions6}
are the condition that $S$, $R$, $U$, and $T$ are Killing spinors, while
\cref{eq:order2AuxiliaryConditions2,%
eq:order2AuxiliaryConditions5,%
eq:order2AuxiliaryConditions7,%
eq:order2AuxiliaryConditions8}
are precisely auxiliary condition A, and 
\cref{eq:order2AuxiliaryConditions9} is just the existence of constant scalar field, so it adds no restrictions.

Fels and Kamran derived in 1990 a subset of the second order symmetry operators for the massive Dirac equaion that can be defined on a curved spacetime~\cite[theorem 4.1]{FelKam90}.
Their ansatz (4.2) for $L$ is less general than \cref{eq:2nd_order_ansatz_lambda,eq:2nd_order_ansatz_gamma} due to a special form of the second order term.
In terms of our covariant language, it can be expressed as
\begin{subequations}
\begin{align}
\lambda ={}&\underset{2, 2}{\mathbb{K} 2}{\overset{2,2}{\odot}}\sTwist \sTwist \phi
 + \underset{3,1}{\mathbb{K} 1}{}{\overset{2,1}{\odot}}\sTwist \phi
 -  \tfrac{2}{3}\underset{1,1}{\mathbb{K} 1}{}{\overset{1,1}{\odot}}\sTwist \phi
 + \underset{2,0}{\mathbb{K} 0}{}{\overset{1,0}{\odot}}\phi
 -  \tfrac{1}{3}\underset{2, 2}{\mathbb{K} 2}{\overset{1,2}{\odot}}\Phi {\overset{1,0}{\odot}}\phi\nonumber\\
& + (2 m^2 \underset{0, 0}{\mathbb{L} 2} - 6 \underset{0, 0}{\mathbb{L} 2} \Lambda -  \tfrac{1}{2} \underset{0,0}{\mathbb{K} 0}{})\phi
 + \underset{2,2}{\mathbb{L} 1}{}{\overset{1,2}{\odot}}\sTwist \chi
 -  \tfrac{1}{2}\underset{0,2}{\mathbb{L} 1}{}{\overset{0,2}{\odot}}\sTwist \chi
 -  \tfrac{2}{3} m\underset{2, 2}{\mathbb{K} 2}{\overset{1,2}{\odot}}\sTwist \chi
 + \underset{1, 1}{\mathbb{L} 0}{\overset{0,1}{\odot}}\chi  \label{eq:lambdaEqFelsKamran},\\
\gamma ={}&\underset{2,2}{\mathbb{M} 1}{}{\overset{2,1}{\odot}}\sTwist \phi
 -  \tfrac{1}{2}\underset{2,0}{\mathbb{M} 1}{}{\overset{2,0}{\odot}}\sTwist \phi
 + \tfrac{2}{3} m\underset{2, 2}{\mathbb{K} 2}{\overset{2,1}{\odot}}\sTwist \phi
 + \underset{1, 1}{\mathbb{M} 0}{\overset{1,0}{\odot}}\phi
 + (2 m^2 \underset{0, 0}{\mathbb{L} 2} - 6 \underset{0, 0}{\mathbb{L} 2} \Lambda -  \tfrac{1}{2} \underset{0,0}{\mathbb{N} 0}{})\chi\nonumber\\
& + \underset{2, 2}{\mathbb{K} 2}{\overset{2,2}{\odot}}\sTwist \sTwist \chi
 + \underset{1,3}{\mathbb{N} 1}{}{\overset{1,2}{\odot}}\sTwist \chi
 -  \tfrac{2}{3}\underset{1,1}{\mathbb{N} 1}{}{\overset{1,1}{\odot}}\sTwist \chi
 + \underset{0,2}{\mathbb{N} 0}{}{\overset{0,1}{\odot}}\chi
 -  \tfrac{1}{3}\underset{2, 2}{\mathbb{K} 2}{\overset{2,1}{\odot}}\Phi {\overset{0,1}{\odot}}\chi  \label{eq:gammaEqFelsKamran},
\end{align}
\end{subequations}
As in \cref{sub:zeroth_order_symmetry_operator}, the underscript indicates the valence numbers for totally symmetric spinors.
These coefficients can then be matched with the ones in our ansatz to obtain a translation from \cref{eq:lambdaEqFelsKamran,eq:gammaEqFelsKamran} to \cref{eq:2nd_order_ansatz_lambda,eq:2nd_order_ansatz_gamma}.
The most immediate part of this translation is
\begin{align}
	W = Y = X = 0,
\end{align}
Hence the symmetry operators presented in \cite{FelKam90} are a special case of the symmetry operators in \cref{thm:condition_for_2nd_order_symmetry_operator_for_dirac_equation}.
We also remark that Fels and Kamran derived commuting operators, which gives stronger conditions than symmetry operators.
For reference, the full translation is available in our \emph{Mathematica} notebook~\cite[]{myGithubRepo}.

\section{Conclusion}%
\label{sec:conclusion}

In conclusion, the problem of finding symmetry operators to the massive Dirac equation is well-suited for applying computer algebra.

While we have found that there are no nontrivial zeroth order symmetry operators, auxiliary condition A and auxiliary condition B are covariant differential equations involving Killing spinors whose solvability are equivalent to the existence of first and second order symmetry operators respectively.
We managed to interpret auxiliary condition A in fairly direct geometrical terms and auxiliary condition B was found to comprise two decoupled systems of equations that reduced to auxiliary condition A in the case of setting the second order coefficients to zero.

\section*{Acknowledgements}

The authors are grateful to Lars Andersson for identifying \cref{eq:order1AuxiliaryConditions4} as a Killing--Yano condition.
Part of this work was done as one of the authors (S.J.) master thesis project at Chalmers University of Technology.

\providecommand{\newblock}{}


%

\end{document}